\documentclass[conference,twocolumn]{IEEEtran}

\usepackage{latexsym}
\usepackage{bm}
\usepackage[dvips]{graphics}
\usepackage{graphicx}
\usepackage{psfrag}
\usepackage{amsfonts,amsmath,amssymb,hyperref}
\usepackage{algorithm,algorithmic}
\usepackage{fixltx2e}
\usepackage{stfloats,cite}

\usepackage{dsfont,color,subfigure}

\usepackage{xspace}
\usepackage{bbm}
\newcommand{\Ac}{\mathcal{A}}

\newcommand{\Dc}{\mathcal{D}}
\newcommand{\Ec}{\mathcal{E}}

\newcommand{\Nc}{\mathcal{N}}

\newcommand{\Uc}{\mathcal{U}}
\newcommand{\Vc}{\mathcal{V}}

\newcommand{\Xc}{\mathcal{X}}
\newcommand{\Yc}{\mathcal{Y}}

\def\a{\alpha}
\def\b{\beta}

\def\e{\epsilon}

\DeclareMathOperator\E{E}
\let\P\relax
\DeclareMathOperator\P{P}
\def\textiid{i.i.d.\@\xspace}
\newcommand\iid{\ifmmode\text{ i.i.d. } \else \textiid \fi}

\newcommand{\ind}{\mathbbmss{1}}

\newtheorem{theorem}{Theorem}
\newtheorem{lemma}{Lemma}

\begin{document}

\title{On the Separation of Lossy Source-Network Coding and Channel Coding in Wireline Networks}

\author{\authorblockN{Shirin Jalali}
\authorblockA{Center for Mathematics of Information\\
California Institute of Technology\\
Pasadena, California, 91125 \\
Email: shirin@caltech.edu}
\and
\authorblockN{Michelle Effros}
\authorblockA{Department of Electrical Engineering\\
California Institute of Technology\\
Pasadena, California, 91125 \\
Email: effros@caltech.edu}
}

\maketitle

\newcommand{\p}{\mathds{P}}
\newcommand{\mb}{\mathbf{m}}
\newcommand{\bb}{\mathbf{b}}
\newcommand{\Xb}{\mathbf{X}}
\newcommand{\Yb}{\mathbf{Y}}
\newcommand{\La}{\Lambda}
\newcommand{\su}{\underline{s}}
\newcommand{\xu}{\underline{x}}
\newcommand{\yu}{\underline{y}}
\newcommand{\Xu}{\underline{X}}
\newcommand{\Yu}{\underline{Y}}
\newcommand{\Uu}{\underline{U}}

\begin{abstract}
This paper proves the separation between source-network coding and channel coding in networks of noisy, discrete, memoryless channels. We show that the set of  achievable distortion matrices in delivering a family of dependent sources across such a network  equals the set of achievable distortion matrices for delivering the same sources across a distinct network which is built by replacing each channel  by a noiseless, point-to-point bit-pipe of the corresponding capacity. Thus a code that applies source-network coding across links that are made almost lossless through the application of independent channel coding across each link asymptotically achieves the optimal performance across the network as a whole.
\end{abstract}

\section{Introduction}
In his seminal work~\cite{Shannon:48},
Shannon separates the problem
of communicating a memoryless source
across a single noisy, memoryless channel
into separate lossless source coding
and channel coding problems.
The corresponding result for lossy coding
in point-to-point channels is almost immediate
since lossy coding in a point-to-point channel
is equivalent to lossless coding
of the codeword indices, and it appears in the same work~\cite{Shannon:48}.
For a single point-to-point channel,
separation holds
under a wide variety of source and channel distributions
(see, for example~\cite{VembuV:95} and the references therein).
Unfortunately, separation does not necessarily hold in network systems.
Even in very small networks
like the multiple access channel~\cite{CoverE:80},
separation can fail when statistical dependencies
between the sources at different network locations
are useful for increasing the rate across the channel.
Since source codes tend to destroy such dependencies,
joint source-channel codes can achieve better performance
than separate source and channel codes in these scenarios.

This paper proves the separation between source-network coding
and channel coding in networks of independent noisy, discrete, memoryless channels (DMC).
Roughly, we show that the vector of achievable distortions
in delivering a family of dependent sources across such a network $\cal N$
equals the vector of achievable distortions
for delivering the same sources across a distinct network $\hat{\cal N}$.
Network $\hat{\cal N}$ is built
by replacing each channel $p(y|x)$ in $\cal N$
by a noiseless, point-to-point bit-pipe
of the corresponding capacity $C=\max_{p(x)}I(X;Y)$.
Thus a code that applies source-network coding
across links that are made almost lossless through the application
of independent channel coding across each link
asymptotically achieves the optimal performance
across the network as a whole.
Note that the operations of network source coding and network coding
are not separable, as shown in \cite{EffrosM:03}  and \cite{RamamoorthyJ:06}
for non-multicast and multicast lossless source coding, respectively.
As a result, a joint network-source code is required, and only the channel code can be separated.
While the achievability of a separated strategy is straightforward,
the converse is more difficult since
preserving statistical dependence between
codewords transmitted across distinct edges of a network of noisy links
improves the end-to-end network performance
in some networks~\cite{KoetterE:09a}.

The results derived here give a partial generalization
of~\cite{Borade:02,SongY:06} and~\cite{KoetterE:09a},
which prove the separation between network coding and channel coding
for multicast \cite{Borade:02,SongY:06} and general demands \cite{KoetterE:09a}, respectively,
under the assumption that
messages transmitted to different subset of users are independent and are  uniformly distributed.
The shift here is from independent sources  to dependent sources,
from lossless to lossy data description,
and from memoryless to non-memoryless sources.
%The proof that these results for lossy coding
%also generalize the earlier results for lossless coding
%relies on the a result from~\cite{GuE:07a,Gu:09},
%which shows that for the proposed family of source coding problems,
%the lossless network source coding region
%equals the lossy network source coding region
%evaluated at distortion zero.
%The generalization is partial
%because~\cite{KoetterE:09a} holds
%for both discrete and continuous channels
%while this paper considers only discrete channels.

The remainder of the paper is organized as follows.
Sections~\ref{sec:not} and~\ref{sec:setup}
describe the notation and problem set-up, respectively.
Section~\ref{sec:stacked_net} describes a tool called a stacked network
that allows us to employ typicality across copies of a network
rather than typicality across time in the arguments that follow.
Section~\ref{sec:results} gives our main results for both
memoryless sources and sources with memory.

\section{Notation}\label{sec:not}
Calligraphic letters, like $\Xc$, $\Yc$, and  $\Uc$, refer to sets, and the size of a set $\Ac$ is denoted by $|\Ac|$. For a random variable $X$, its alphabet set is represented by $\Xc$.

While a random variable is denoted by $X$, $\Xu$ represents a random vector. The length of a vector is implied in the context, and its $\ell^{\rm th}$ element is denoted by $\Xu(\ell)$.

For two vectors $\xu_1$ and $\xu_2$ of the same length $r$, $\|\xu_1-\xu_2\|_1$ denotes the $\ell_1$ distance between the two vectors defined as $\|\xu_1-\xu_2\|_1 = \sum\limits_{i=1}^r|\xu_1(i)-\xu_2(i)|$.
If $\xu_1$ and $\xu_2$ represent probability distributions, i.e., $\sum\limits_{i=1}^r\xu_1(i)=\sum\limits_{i=1}^r\xu_2(i)=1$ and $\xu_1(i), \xu_2(i)\geq 0$ for all $i\in\{1,\ldots,r\}$, then the total variation distance between $\xu_1$ and $\xu_2$ is defined as $ \|\xu_1-\xu_2\|_{\rm TV}=0.5\|\xu_1-\xu_2\|_1.$

Unlike \cite{KoetterE:09a}, this paper uses strong typicality arguments to demonstrate the equivalence between noisy channels and noiseless bit-pipes of the same capacity.  We therefore assume that the channel input and output alphabets are finite.  The alphabets for the sources described across the channel may be discrete or continuous.

\section{The problem setup}\label{sec:setup}
Consider a multiterminal network $\Nc$ consisting of $m$ nodes  interconnected via some point-to-point, independent DMCs. The network structure is represented by a directed graph $G$ with node set $\Vc = \{1,\ldots,m\}$ and edge set $\Ec$. Each directed edge $e=[v_1,v_2]\in\Ec$ implies a  point-to-point DMC between nodes $v_1$ (input) and $v_2$ (output).
Each node $a$  observes some source process $\mathbf{U}^{(a)}=\{U^{(a)}_k\}_{k=1}^{\infty}$, and is interested in reconstructing a subset of the processes observed by the other nodes. The alphabet of source $\mathbf{U}^{(a)}$, $\Uc^{(a)}$,  can be either scalar or vector-valued. This allows  node $a$ to have a vector of sources. For achieving this goal in a block coding framework, source output symbols are divided into non-overlapping blocks of length $L$. Each block  is described separately. At the beginning of the $j^{\rm th}$ coding period, each node $a$ has observed a length-$L$ block of the process $\mathbf{U}^{(a)}$, i.e., $U^{(a),j L}_{(j-1)L+1}=(U^{(a)}_{(j-1)L+1},\ldots,U^{(a)}_{j L})$. The blocks $\{U^{(a),j L}_{(j-1)L+1}\}_{a\in\Vc}$ observed at different nodes are described over the network in $n$ uses of the network (The rate $\kappa\triangleq\frac{L}{n}$ is a parameter of the code). For those $n$ time steps, at each step $t\in\{1,\ldots,n\}$, each node $a$  generates its next channel inputs as a function of $U^{(a),j L}_{(j-1)L+1}$ and its channels' outputs up to time $t-1$, here denoted by $Y^{(a),t-1}=(Y^{(a)}_1,\ldots,Y^{(a)}_{t-1})$, according to
\begin{align}
X_t^{(a)}:(\Yc^{(a)})^{t-1}\times\Uc^{(a),L}\to\Xc^{(a)}.\label{eq:Xt}
\end{align}
Note that each node might be the input to more than one channel and/or the output of more than one channel. Hence, both $X_t^{(a)}$ and $Y_t^{(a)}$ might be vectors depending on the indegree and outdegree of node $a$. The reconstruction at node $b$ of the block observed at node $a$ is denoted by $\hat{U}^{(a\to b),L}$. This reconstruction is a function of the source observed at node $b$ and node $b$'s channel outputs, i.e., $\hat{U}^{(a\to b),L}= \hat{U}^{(a\to b)}(Y^{(b),n},U^{(b),L})$, where
\begin{align}
\hat{U}^{(a\to b)}:(\Yc^{(b)})^{n}\times\Uc^{(a),L}\to\hat{\Uc}^{(a\to b),L}.\label{eq:Uhat_t}
\end{align}

The performance criterion for a coding scheme is its induced expected average distortions between sources and reconstruction blocks, i.e., for all $a,b\in\Vc$
\begin{align*}
\E d^{(a\to b)}_L(U^{(a),L},\hat{U}^{(a\to b),L}) \triangleq \E \frac{1}{L}\sum\limits_{k=1}^Ld^{(a\to b)}(U^{(a)}_k,\hat{U}^{(a\to b)}_k),
\end{align*}
where $d^{(a\to b)}:\Uc^{(a)}\times\hat{\Uc}^{(a\to b)}\rightarrow {\mathds{R}}^+$ is a per-letter distortion measure. As mentioned before $\Uc^{(a)}$ and $\hat{\Uc}^{(a\to b)}$ are either scalar or vector-valued. This allows the case where node $a$ observes multiple sources and node $b$ is interested in reconstructing a subset of them.  Let 
\[
d_{\max}\triangleq\max\limits_{a,b,\in \Vc, \a \in \Uc^{(a)},\b\in\hat{\Uc}^{(a\to b)}} d^{(a\to b)}(\a,\b)<\infty.
\]
If node $b$ is not interested in reconstructing node $a$, then  $d^{(a\to b)}\equiv 0$. 

The distortion matrix $\mathbf{D}$ is said to be achievable at a rate $\kappa$ in a network $\Nc$, if for any $\e>0$, there exists a pair $(L,n)$, $L/n=\kappa$, and  block length $n$ coding scheme such that
\begin{align}
\E d^{(a \to b)}_L(U^{(a),L},\hat{U}^{(a\to b),L}) \leq D(a,b) + \e,\label{eq:dist}
\end{align}
for any $a,b\in \Vc$.

\section{Stacked network} \label{sec:stacked_net}
For a given network $\Nc$, the corresponding $N$-fold stacked network $\underline{\Nc}$ is defined as $N$ copies of the original network \cite{KoetterE:09a}. That is, for each node and each edge in $\Nc$, there are $N$ copies of the same node or  same edge in  $\underline{\Nc}$. At each time instance, each node has access to the data available at nodes which are its copies, and potentially uses this extra information in generating the channel inputs of the future time instances. Likewise, in decoding, all $N$ copies of a node can collaborate in reconstructing the signals. This is made more precise in the following two definitions
\begin{align}
\Xu_t^{(a)}:(\underline{\Yc}^{(a)})^{t-1}\times\Uc^{(a),NL}\to\underline{\Xc}^{(a)},\label{eq:Xt_ul}
\end{align}
and
\begin{align}
\hat{\Uu}^{(a\to b)NL}:(\underline{\Yc}^{(b)})^{n}\times\Uc^{(b),NL}\to\hat{\Uc}^{(a\to b),NL},\label{eq:Uhat_t_ul}
\end{align}
which correspond to \eqref{eq:Xt} and \eqref{eq:Uhat_t} in the original network. In \eqref{eq:Xt_ul} and \eqref{eq:Uhat_t_ul} all the vectors are of length $N$.

In an $N$-layered network, the distortion between the source observed at node $a$ and its reconstruction at node $b$ is defined as
\begin{align}
D_N(a,b)=\E \left[d^{(a\to b)}_{NL}(U^{(a\to b),NL},\hat{U}^{(a\to b),NL})\right],
\end{align}
for any $a,b\in\{1,\ldots,m\}$.

A distortion matrix $\mathbf{D}$ is said to be achievable in the stacked network at some rate $\kappa$ if for any given $\e>0$, there exist $N$, $n$ and $L$ large enough, such that $D_N(a,b) \leq D(a,b)+\e,$ for all $a,b\in\{1,\ldots,m\}$.
Note that the dimension of the distortion matrices in both single layer and multi-layer networks is $m\times m$. Let $\Dc(\kappa,\Nc)$ and  $\Dc_{s}(\kappa,\underline{\Nc})$ denote the closure of the set of achievable distortion matrices at some rate $\kappa$ in a network $\Nc$ and its stacked version $\underline{\Nc}$ respectively. The following theorem establishes the relationship between the two sets.

\begin{theorem} \label{thm:single_eq_multi}
At any rate $\kappa$, 
\begin{align}
\Dc(\kappa,\Nc)=\Dc_{s}(\kappa,\underline{\Nc}).
\end{align}
\end{theorem}
\begin{proof}
\begin{itemize}

\item[i.] Proof of $\Dc(\kappa,\Nc) \subseteq \Dc_s(\kappa,\underline{\Nc})$. Consider any $\mathbf{D}\in {\rm int}(\Dc(\kappa,\Nc))$. Then for any $\e>0$, there exists a coding operating scheme at rate $\kappa=L/n$ on $\Nc$ such that \eqref{eq:dist} is satisfied. For any $N$, a stacked network that uses this same coding strategy independently in each layer achieves
    \begin{align}
    \E [ d^{(a\to b)}_{NL}&(U^{(a\to b),NL},\hat{U}^{(a\to b),NL})] \nonumber\\
    &=\frac{1}{N} \sum\limits_{\ell=1}^N \E [d^{(a\to b)}_{L}(U^{(a\to b),\ell L}_{(\ell-1)L+1},\hat{U}^{(a\to b),\ell L}_{(\ell-1)L+1})]\nonumber\\
    &\leq\frac{1}{N}\sum\limits_{\ell=1}^N D(a,b)+\e\nonumber\\
    &=D(a,b)+\e.
    \end{align}

\item[ii.] $\Dc_s(\kappa,\underline{\Nc}) \subseteq \Dc(\kappa,\Nc)$. Let $\mathbf{D}\in {\rm int} (\Dc_s(\kappa,\underline{\Nc}))$. Since $\mathbf{D}\in{\rm int}(\Dc_{s}(\kappa,\underline{\Nc}))$, for any $\e>0$, there exists integers $N$, $n$, and $L$ such that a stacked network consisting of $N$ layers along with a block length $n$ coding scheme for $L$ source symbols on this stacked network achieves 
    \[
    \E \left[d^{(a\to b)}_{NL}(U^{(a\to b),NL},\hat{U}^{(a\to b),NL})\right] \leq D(a,b) + \e,
    \]
 for all $a,b\in\Vc$. The same coding scheme can be used in a single-layer network as follows. Consider a single layer network where each node observes a length-$NL$ block of source symbols and describes the block in the next $Nn$ time steps. At times $t\in\{1,\ldots,N\}$, each node $a$ sends what would have been sent at time 1 by node $a$ in layer $t$ of the stacked network. After that, having collected the output of the previous $N$ time steps, at times $t\in\{N+1,\ldots,2N\}$, node $a$ sends the outputs of the same node at time 2 in layer $t-N$ (Note that in the first $N$ time steps, node $a$'s output is only a function of its own source, not the channels' outputs. It only collects the channel outputs in order to use them during the next $N$ time steps.). The same strategy is used in $n$ time intervals, each comprising $N$ network uses. During each period, the new channel outputs observed by node $a$ are recorded to be used in the future periods, but do not affect the next inputs generated by that node during that time period. Using this strategy, at the end of $nN$ channel uses, each node's observation has exactly the same distribution as the collection of observations of its $N$ copies in the stacked networks. Therefore, applying the same decoding rule will result in the same performance. Hence, $\mathbf{D}\in\Dc(\kappa,\Nc)$.

\end{itemize}
\end{proof}

\section{Replacing a noisy channel with a bit pipe} \label{sec:results}
\subsection{Memoryless sources}

In this section we assume all sources are jointly i.i.d., i.e., for any $k\geq1$, $\P(U^{(1),k},\ldots,U^{(m),k})=\prod\limits_{i=1}^k\P(U^{(1)}_i,\ldots,U^{(m)}_i)$, where $\P(U^{(1)}_i,\ldots,U^{(m)}_i)$ does not depend on $i$. Note that at each time instant the sources might be correlated with each other.

%\begin{align}
%\P(U^{(1),k},\ldots,U^{(m),k})=\prod\limits_{i=1}^k\P(U^{(1)}_i,\ldots,U^{(m)}_i),
%\end{align}
In the described network $\Nc$, for some $a,b\in\Vc$ such that $[a,b]\in\Ec$, consider the noisy channel connecting these two nodes. The channel is described by its transition probabilities $\{p(y|x)\}_{x\in\Xc,y\in\Yc}$, and has some finite capacity $C=\max\limits_{p(x)}I(X;Y).$ Now consider a network $\Nc'$ which is identical to $\Nc$ except for the noisy channel between $a$ and $b$, which is replaced by a bit-pipe of capacity $C$. 
\begin{theorem}\label{thm:main}
For any $\kappa>0$,
\begin{align}
\Dc(\kappa,\Nc)=\Dc(\kappa,\Nc').
\end{align}
\end{theorem}
\begin{proof}[Proof outline]
By Theorem \ref{thm:single_eq_multi}, the achievable region of a network is equal to the achievable region of its stacked version. Hence, it suffices to prove that $\Dc_s(\kappa,\underline{\Nc})=\Dc_s(\kappa,\underline{\Nc'})$.
\begin{itemize}
\item[i.] $\Dc_s(\kappa,\underline{\Nc}')\subseteq\Dc_s(\kappa,\underline{\Nc})$: Let $\mathbf{D}\in{\rm int}(\Dc_s(\kappa,\underline{\Nc}'))$. We need to show that $\mathbf{D}\in\Dc_s(\kappa,\underline{\Nc})$ as well. Note that ${\Nc}$ and $\Nc'$ are identical except for the DMC connecting nodes $a$ and $b$ in $\Nc$ which is replaced by a bit-pipe of capacity $C$ in $\Nc'$. We next show that  any code for $\underline{\Nc}'$ can be  operated on $\underline{\Nc}$  with a similar expected distortion. Let the number of layers  in both networks be $N$. Given the capacity of the bit-pipes, the number of bits that can be carried from $a$ to $b$ in $\underline{\Nc}'$ is at most $NR$, where $R<C$. Hence, if $N$ is large enough, the same information can be transmitted from $a$ to $b$ in $\underline{\Nc}$ by doing appropriate channel coding across the layers over the noisy channel and its copies connecting $a$ and $b$ in $\underline{\Nc}$. Let $P_{e,(a\to b)}$ denote the probability of error of the channel code operating over the channel corresponding to the edge $[a,b]$ and its copies in  $\underline{\Nc}$, and let $P_{e,\max}=\max_{[a,b]\in\Ec}P_{e,a\to b}$.  Then the extra expected distortion introduced at each reconstruction point is bounded above by $|\Ec|P_{e,\max}d_{\max}$ and can be made arbitrarily small.

\item[ii.] $\Dc(\kappa,\underline{\Nc})\subseteq\Dc_s(\kappa,\underline{\Nc'})$: Let $\mathbf{D}\in {\rm int}(\Dc(\kappa,\Nc))$. We prove that $\mathbf{D}\in\Dc_s(\kappa,\underline{\Nc}')$. Consider a code defined on $\Nc$ that achieves within $\e$ of $\mathbf{D}$, and consider the $N$-fold stacked version of $\Nc$, $\underline{\Nc}$. Assume that the same code is applied independently in each layer. We first show that, when all sources are memoryless and uniformly distributed, the performance of the code given the realization of  $(\underline{X}_1,\underline{Y}_1)$ only depends on the empirical distribution of $(\underline{X}_1,\underline{Y}_1)$ defined as
\begin{align}
\hat{p}_{[\underline{X}_1,\underline{Y}_1]}(x,y)=\frac{1}{N}\sum\limits_{\ell=1}^N\ind_{(\underline{X}_1(\ell),\underline{Y}_1(\ell))=(x,y)},
\end{align}
for all $x\in\Xc$ and $y\in\Yc$. Here the subscript $1$ refers to time $t=1$. After establishing this, we use the result proved in \cite{CuffP:09} and show that at time $t=1$ we can simulate the performance of the noisy link by a bit-pipe of the same capacity. For the rest of the proof, let $U=\{U_i\}$ and $\hat{U}=\{\hat{U}_i\}$ denote some i.i.d. source observed at some node in $\Vc$ and its reconstruction at some other node in $\Vc$.

In the original network,
\begin{align}
&\E d_L(U^{L},\hat{U}^{L}) = \sum_{\substack{
              x\in\Xc \\
              y\in\Yc }} \E\left[d_L(U^{L},\hat{U}^{ L})\left|(X_1,Y_1)=(x,y)\right.\right]\nonumber\\
&\hspace{2.8cm}\times\P\left((X_1,Y_1)=(x,y)\right).\label{eq:single_layer_t1}
\end{align}

On the other hand, in the $N$-fold stacked network,
\begin{align}
&\E \left[d_{NL}(U^{NL},\hat{U}^{NL})\right] \nonumber\\
&=\E\left[\sum\limits_{\ell=1}^N\frac{d_{L}\left(U^{\ell L}_{(\ell-1)L+1},\hat{U}^{\ell L}_{(\ell-1)L+1}\right)}{N} \times\right. \nonumber\\
&\hspace{1cm} \left.\sum_{\substack{
              x\in\Xc \\
              y\in\Yc }} \ind_{(\underline{X}_1(\ell),\underline{Y}_1(\ell))=(x,y)} \right]\nonumber\\
&=\E \left[\sum_{\substack{
              x\in\Xc \\
              y\in\Yc }}\sum\limits_{\ell=1}^N\right.\nonumber\\
&\left.\frac{d_{L}\left(U^{\ell L}_{(\ell-1)L+1},\hat{U}^{\ell L}_{(\ell-1)L+1}\right) \ind_{(\underline{X}_1(\ell),\underline{Y}_1(\ell))=(x,y)}}{N}  \right]\nonumber\\
&= \sum_{\substack{
              x\in\Xc \\
              y\in\Yc }} \E\left[d_L(U^{L},\hat{U}^{L})\left|(X_1,Y_1)=(x,y)\right.\right]\times\nonumber\\
&\hspace{1 cm} \E[ \hat{p}_{[\Xu_1,\Yu_1]}(x,y)].\label{eq:layer_t1}
\end{align}

Comparing \eqref{eq:single_layer_t1} and \eqref{eq:layer_t1} reveals that the desired result will follow if we can find a coding scheme for which,
\begin{align}
\left|\P\left((X_1,Y_1)=(x,y)\right)-\E[ \hat{p}_{[\Xu_1,\Yu_1]}(x,y)]\right|,
\end{align}
can be made arbitrary small.

To prove this, consider a channel with input  drawn i.i.d. from some distribution $p(x)$. The encoder observes $N$ source symbols and sends a message of $NR$ bits to the decoder. The decoder converts these $NR$ bits into a reconstruction block $\Yu=(Y_1,\ldots,Y_N)$. The empirical joint distribution between the channel input and channel output induced by the bit pipe is defined in the classical sense as follows
\[
\hat{p}_{[\Xu,\Yu]}(x,y)=\frac{1}{N}\sum\limits_{\ell=1}^N\ind_{(\Xu(\ell),\Yu(\ell))=(x,y)}.
\]
Consider a DMC described by  transition probabilities  $\{p(y|x)\}_{x\in\Xc,y\in\Yc}$ whose input is an i.i.d. process distributed according to some distribution $p(x)$. In \cite{CuffP:09}, it is shown that, as long as $R>I(X;Y)$, any such channel  can be simulated by a  bit pipe of rate at most $R$ such that the total variation between $\hat{p}_{[\Xu,\Yu]}(x,y)$ and $p(x,y)=p(x)p(y|x)$ can be made arbitrarily small for large enough block lengths. In other words, there exists a sequence of coding schemes over the bit-pipe such that
\begin{align}
\left\|\hat{p}_{[\Xu,\Yu]}-p\right\|_1\stackrel{n\to\infty}{\longrightarrow} 0\;{\rm a.s.}
\end{align}
(where $\hat{p}_{[\Xu,\Yu]}$ and $p$ are vectors describing distributions ($\hat{p}_{[\Xu,\Yu]}(x,y): x,y\in\Xc,\Yc$) and ($p(x,y): x,y\in\Xc,\Yc$) respectively.) 

Combining this result with our initial claim yields the desired result, i.e., at time $t=1$, we can replace the noisy link by a bit-pipe. To extend this result to the next $n-1$ time steps, we use induction. Note that in the original network
\begin{align}
&\E d_L(U^{L},\hat{U}^{L}) =\nonumber\\
& \sum_{\substack{
              x_t\in\Xc \\
              y_t\in\Yc \\
              t=1,\ldots,n }} \E\left[d_L(U^{L},\hat{U}^{ L})\left|\bigcap\limits_{t=1}^n\{(X_t,Y_t)=(x_t,y_t)\}\right.\right]\nonumber\\
&\hspace{1cm}\times\P\left((X^n,Y^n)=(x^n,y^n)\right).
\end{align}

On the other hand, using the same analysis used in deriving \eqref{eq:layer_t1}, in the $N$-fold stacked network,
\begin{align}
&\E \left[d_{NL}(U^{NL},\hat{U}^{NL})\right] \times\nonumber\\
& = \sum_{\substack{
              x_t\in\Xc \\
              y_t\in\Yc \\
              t=1,\ldots,n }} \E\left[d_L(U^{L},\hat{U}^{L})\left|(X^n,Y^n)=(x^n,y^n)\right.\right]\nonumber\\
&\hspace{1cm}\times \E\left[ \frac{\left|\left\{\ell:(\Xu^t(\ell),\Yu^t(\ell))=(x^t,y^t)\right\}\right|}{L}\right].
\end{align}
Therefore, we need to show that by appropriate coding over the bit-pipes,
\begin{align}
&\left|\P\left((X^n,Y^n)=(x^n,y^n)\right)\textcolor{white}{ \frac{\Xu^t(\ell)}{L}}\right.\nonumber\\
&\left.\hspace{1cm}- \E\left[ \frac{\left|\left\{\ell:(\Xu^t(\ell),\Yu^t(\ell))=(x^t,y^t)\right\}\right|}{L}\right]\right|\label{eq:p_e}
\end{align}
can be made arbitrary small.  Note that
\begin{align}
&\P\left((X^n,Y^n)=(x^n,y^n)\right)=\prod\limits_{t=1}^n\nonumber\\
& \P\left((X_t,Y_t)=(x_t,y_t)\left|(X^{t-1},Y^{t-1})=(x^{t-1},y^{t-1})\right.\right),\label{eq:single_layer_t_all}
\end{align}
and 
\begin{align}
&\frac{\left|\left\{\ell:(\Xu^n(\ell),\Yu^n(\ell))=(x^n,y^n)\right\}\right|}{L}\nonumber\\
&=\prod\limits_{t=1}^n\frac{\left|\left\{\ell:(\Xu^t(\ell),\Yu^t(\ell))=(x^t,y^t)\right\}\right|}
{\left|\left\{\ell:(\Xu^{t-1}(\ell),\Yu^{t-1}(\ell))=(x^{t-1},y^{t-1})\right\}\right|},\label{eq:layer_t_all}
\end{align}
where for $t=1$ 
\[
\left|\left\{\ell:(\Xu^{t-1}(\ell),\Yu^{t-1}(\ell))=(x^{t-1},y^{t-1})\right\}\right|=L.
\]
We have already proved that by appropriate coding, we can make the first term in \eqref{eq:layer_t_all}  converge to the first term in \eqref{eq:single_layer_t_all}  with probability one. By induction, we can prove that the same result is true for any other term in  \eqref{eq:layer_t_all} and its corresponding term  in \eqref{eq:single_layer_t_all}.  After proving this, since all the terms in \eqref{eq:layer_t_all} and as a result their product are positive and  upper-bounded by $1$, we can use the Dominated Convergence Theorem (see, for example, \cite{Durett:91}) to show that \eqref{eq:p_e} can be made arbitrary small.

To apply induction, assume there exist some coding schemes by which  we make the first $t-1$ terms in \eqref{eq:layer_t_all} each converge to the corresponding term in \eqref{eq:single_layer_t_all}  almost surely.  Using this assumption, we prove that the same thing is true for the $t^{\rm th}$ term as well.

Note that when the first $t-1$ terms are very close,  the frequency of occurrence of each pattern $\{(\Xu^{t-1}(\ell),\Yu^{t-1}(\ell))=(x^{t-1},y^{t-1})\}$ across the layers in $\underline{\Nc}$  is very close to the pattern's probability. Since the two networks perform the same except for link $[a,b]$, the network guarantees that the frequency of $\{(\Xu^{t}(\ell),\Yu^{t-1}(\ell))=(x^{t},y^{t-1})\}$ is also close to its probability in $\underline{\Nc}$. In order to finish the proof, we use Lemma \ref{lemma:1} proved in Appendix 1. 
\begin{lemma}\label{lemma:1}
If we choose the random codes used at times  $t-1$ and $t$ independently, then
\begin{align}
\E&[\ind_{\Yu_t(1)=y_t}|(\Xu_{t-1}(1),\Yu_{t-1}(1))=(x_{t-1},y_{t-1}),\nonumber\\
& \Xu_t(1)=x_t]= \P\left(\Yu_t(1)=y_t|\Xu_t(1)=x_t\right),
\end{align}
where the expectation is both with respect to the network and the code selections.
\end{lemma}

\end{itemize}
\end{proof}

\subsection{Sources with memory}
Assume that the sources are no longer memoryless but mixing. That is for any integers $k$ and $T$
\begin{align*}
&\left|\P\left((U^{(1),k},\ldots,U^{(m),k},U_T^{(1),T+k},\ldots,U_T^{(m),T+k}) = \right.\right.\nonumber\\
&\hspace{0.5cm}\left.(u^{(1),k},\ldots,u^{(m),k},u_T^{(1),T+k},\ldots,u_T^{(m),T+k}) \right) - \nonumber\\
& \P\left((U^{(1),k},\ldots,U^{(m),k}) = (u^{(1),k},\ldots,u^{(m),k})\right)\times \nonumber\\
&\left.\P\left((U_T^{(1),T+k},\ldots,U_T^{(m),T+k}) =  (u_T^{(1),T+k},\ldots,u_T^{(m),T+k})\right)\right|
\end{align*}
goes to $0$ as $T$ approaches $\infty$. In the proof of Theorem \ref{thm:main}, we used the fact that the sources are correlated and jointly i.i.d.  to conclude that the inputs to the copies of a channel  in the stacked network are i.i.d.   If the sources have memory, this does not hold any more. But, if we assume that the sources are mixing, then for block length $L$  large enough,  the two sets $\{U^L,U_{2L+1}^{3L},\ldots\}$ and $\{U_{L+1}^{2L},U_{3L+1}^{4L},\ldots\}$ look like two i.i.d. sequences.  Therefore, in the stacked network, if we code the even-numbered layers together and the odd-numbered ones together, such that each one is done separate from the other one, we  get back to the i.i.d. regime and can prove a similar result.

%------------------Appendix I -----------------
\renewcommand{\theequation}{A-\arabic{equation}}
% redefine the command that creates the equation no.
\setcounter{equation}{0}  % reset counter

\section*{Appendix A: Proof of Lemma \ref{lemma:1}}
Note that
\begin{align}
\E&\left[\ind_{\Yu_t(1)=y_t }|\Xu_{t-1}(1)=x_{t-1},\Yu_{t-1}(1)=y_{t-1},\Xu_t(1)=x_t \right]\nonumber\\
=&\sum\limits_{\su_1,\hat{\su}_1,\su_2}\P\left(\underline{Y}_t(1) =y_t ,\Xu_{t-1}(2:N)=\su_1,\right.\nonumber\\
&\hspace{1cm}\Yu_{t-1}(2:N)=\hat{\su}_1,\Xu_t(2:N)=\su_2|\Xu_{t-1}(1)=x_{t-1},\nonumber\\
&\hspace{1cm}\left.\Yu_{t-1}(1)=y_{t-1},\Xu_t(1)=x_t \right)\nonumber\\
=&\sum\limits_{\su_2}\P(\Yu_t(1)=y_t |\Xu_t=[x_t ,\su_2])\P(\Xu_t(2:N)=\su_2|\nonumber\\
&\Xu_{t-1}(1)=x_{t-1},\Yu_{t-1}(1)=y_{t-1},\Xu_t(1)=x_t ).\label{eq:a1}
\end{align}
But
\begin{align}
&\P(\Yu_t(1)=b |\Xu_t=\xu_t) \nonumber\\
&= \sum\limits_{\yu_t:\yu_t(1)=b }\frac{\P(\Xu_t=\xu_t,\Yu_t=\yu_t)}{\P(\Xu_t=\xu_t)}\\
& = \frac{1}{\P(\Xu_t=\xu_t)}\sum\limits_{\Yu_t:\Yu_t(1)=b }\P(\Xu_t(1)=\xu_t(1))\times \nonumber\\
&\P(\Yu_t(1)=y_t |\Xu_t(1)=\xu_t(1))\times\nonumber\\
&\P(\Xu_t(2:N)=\xu_t(2:N)|\Xu_t(1)=\xu_t(1),\Yu_t(1)=b)\times\nonumber\\
&\P(\Yu_t(2:N)=\yu_t(2:N)|\Xu_t=\xu_t,\Yu_t(1)=\yu_t(1))\\
& = \frac{1}{\P(\Xu_t=\xu_t)}
\sum\limits_{\yu_t:\yu_t(1)=b }\P(\Xu_t(1)=\xu_t(1))\times \nonumber\\
&\P(\Yu_t(1)=y_t |\Xu_t(1)=\xu_t(1))\times\nonumber\\
&\P(\Xu_t(2:N)=\xu_t(2:N)|\Xu_t(1)=\xu_t(1))\times\nonumber\\
&\P(\Yu_t(2:N)=\yu_t(2:N)|\Xu_t=\xu_t,\yu_t(1)=b)\\
& = \P(\Yu_t(1)=b |\Xu_t(1)=\xu_t(1))\times \nonumber\\
&\sum\limits_{\yu_t(1)=y_t }\P(\Yu_t(2:N)=\yu_t(2:N)|\Xu_t=\xu_t,\Yu_t(1)=b)\nonumber\\
& = \P(\Yu_t(1)=b |\Xu_t(1)=\xu_t(1)).\label{eq:a2}
\end{align}
Combining \eqref{eq:a1} and \eqref{eq:a2} yields the desired result.

\section*{Acknowledgments}
SJ is supported by the Center for Mathematics of Information at Caltech, and ME is supported by the DARPA ITMANET program under grant number W911NF-07-1-0029.

\end{document}